\documentclass[conference]{IEEEtran}
\usepackage{cite}
\usepackage{amsmath,amssymb,amsfonts}
\usepackage{algorithmic}
\usepackage{graphicx}
\usepackage{textcomp}
\usepackage{xcolor}
\usepackage{amsthm}
\usepackage{mathtools}
\usepackage[shortlabels]{enumitem}
\usepackage[hidelinks]{hyperref}
\usepackage[english]{babel}
\def\BibTeX{{\rm B\kern-.05em{\sc i\kern-.025em b}\kern-.08em
    T\kern-.1667em\lower.7ex\hbox{E}\kern-.125emX}}
    
\newtheorem{theorem}{Theorem}[section]
\newtheorem{lemma}[theorem]{Lemma}
\newtheorem{remark}[theorem]{Remark}
\newtheorem{prop}[theorem]{Proposition}
\newtheorem{defn}[theorem]{Definition}
\newtheorem{example}[theorem]{Example}

\newcommand{\arr}[1]{\mathbf{#1}}
\newcommand{\farray}[1]{\mathcal{#1}}
\newcommand{\arrprod}[2]{\mathbin{{#1}.{#2}}}
\newcommand{\tensor}{\mathbin{\text{\textcircled{$\otimes$}}}}
\newcommand{\colo}{{\colon\!\!}}
\DeclareMathOperator{\dom}{dom}
\DeclareMathOperator{\ran}{ran}

\begin{document}

\title{Fuzzy Relational Databases via Associative Arrays
}

\author{\IEEEauthorblockN{Kevin Min, Hayden Jananthan, Jeremy Kepner}
\IEEEauthorblockA{\textit{Massachusetts Institute of Technology}}}

\IEEEoverridecommandlockouts
\IEEEpubid{\makebox[\columnwidth]{979-8-3503-0965-2/23/\$31.00 \copyright 2023 IEEE \hfill}
\hspace{\columnsep}\makebox[\columnwidth]{ }}

\maketitle

\begin{abstract}
    The increasing rise in artificial intelligence has made the use of imprecise language in computer programs like ChatGPT more prominent. Fuzzy logic addresses this form of imprecise language by introducing the concept of fuzzy sets, where elements belong to the set with a certain membership value (called the fuzzy value). This paper combines fuzzy data with relational algebra to provide the mathematical foundation for a fuzzy database querying language, describing various useful operations in the language of linear algebra and multiset operations, in addition to rigorously proving key identities. 
\end{abstract}

\section{Introduction}\label{Intro}

\let\thefootnote\relax\footnotetext{Research was sponsored by the United States Air Force Research Laboratory and the Department of the Air Force Artificial Intelligence Accelerator and was accomplished under Cooperative Agreement Number FA8750-19-2-1000. The views and conclusions
contained in this document are those of the authors and should not be interpreted as representing the official policies, either expressed or implied, of the Department of the Air Force or the U.S. Government. The U.S. Government is authorized to reproduce and distribute reprints for Government purposes notwithstanding any copyright notation herein.}

Ever since their introduction by Codd \cite{10.1145/362384.362685} in 1970, relational database management systems (RDBMSs) such as the Structured Query Language (SQL) have been used to efficiently perform searches and categorization from a variety of fields, including financial systems, manufacturing, social media, and perhaps most recently, artificial intelligence \cite{10.1145/166635.166656}. With the ever-increasing popularity of this last topic comes a need for a system more flexible than traditional relational databases, which are often limited in the type of data they can handle, involving only exact data points. 

Several papers \cite{366566, SOSNOWSKI199023, MEDINA199487, 4052712, LiteratureOverview} have attempted to rectify this limitation by introducing fuzzy sets to databases, to create fuzzy relational databases. Fuzzy sets, introduced by Zadeh \cite{ZADEH1965338, LEE1969421}, assign a degree of partial membership to each of its elements (``fuzzy value"), providing a more general framework than traditional sets, where membership is binary.
Fuzzy relational databases take advantage of this generalization by treating the database as a fuzzy set of its rows, in contrast to traditional databases which are akin to standard sets consisting of its rows.

These fuzzy databases offers several major advantages over traditional RDBMSs, handling data with vagueness or imprecision regarding what entries should be contained in the database. Most prominently, this features in artificial intelligence, where it is often necessary to account for the imprecise semantics of human language \cite{FuzzyAI, electronics10222878}. For example, if one wished to query for ``tall'' employees at a company using a traditional database, an arbitrary cutoff height would be needed such that ``tall'' applies only to those at or above that height. Of course, this doesn't appropriately represent reality; a discrete cutoff does not fully capture the imprecise nature of the descriptor ``tall".

On the other hand, with fuzzy databases one could make a ``fuzzy query" representing the term ``tall". The resulting database from this fuzzy query would associate high fuzzy membership values with the entries corresponding to tall people, and lower fuzzy membership values with the entries corresponding to shorter people. This would provide a more accurate and natural description of the imprecise query.

In this paper we provide a rigorous mathematical formulation of such a fuzzy relational database system, analogous to what \cite{8258298} describes for standard RDBMSs. We use the concept of associative arrays, a generalization of matrices that allows for arbitrary row and column indices instead of the traditional natural number indexing of matrices. By using associative array equivalents of traditional linear algebra operations alongside fuzzy multiset theory \cite{RIESGO2018107, UnionIntersections, inproceedings}, we present a natural method for representing and implementing a fuzzy RDBMS. In \S \ref{Prelim} we cover the relevant technical background;  in \S \ref{FArrays} we introduce the concept of fuzzy arrays; and in \S \ref{Relations} and \S \ref{Relations2} we define key relational algebra operations generalized to fuzzy logic, and prove useful properties involving them.

\section{Preliminaries and Background}\label{Prelim}

\IEEEpubidadjcol

In this section we discuss the mathematical preliminaries to understanding fuzzy relational databases. We focus on fuzzy multisets and associative arrays, which represent the core components of fuzzy databases.
 
\subsection{Multisets}

At its core, databases can be interpreted as multisets of their rows. To represent fuzzy databases, we discuss multisets of fuzzy values. Using this, we can effectively represent the fuzzy database as a fuzzy multiset of its rows, in which rows belong to the fuzzy database with a fuzzy value.

\begin{defn}[multiset of fuzzy values]
    A \emph{multiset of fuzzy values} $S$ (or a \emph{multiset}) is a function $\mathrm{Count}_S \colon [0, 1] \rightarrow \mathbb{N}_{\geq 0}$. 
    
    We say $S$ and $T$ are \emph{equivalent} (denoted $S \equiv T$) if $\mathrm{Count}_S(x)=\mathrm{Count}_T(x)$ for all $x\in (0, 1]$.
    
    The \emph{cardinality} of $S$ is $|S|\coloneqq \sum_{x\in(0, 1]}\mathrm{Count}_S(x)$. Note that cardinality is preserved up to equivalence.
    
    A multiset is \emph{finite} if the cardinality is finite.
    
    We will also use the notation $S=\{a_1, a_2, \dots, a_n\}$ to denote a finite multiset. Here, $\mathrm{Count}_S(x)$ is defined as the number of occurrences of $x$ among $a_1, a_2, \dots, a_n$. Notably, the order of the $a_i$ values does not matter with this notation.
\end{defn}

Throughout this paper, ``multiset'' means ``finite multiset of fuzzy values'' unless otherwise stated.

\begin{defn}[$k$-th element of multiset]
    Given a multiset $M$, define $M[k]$ as the $k$th largest element of the multiset (including multiplicity). If $k>|M|$, define $M[k]=0$.
\end{defn}

\begin{defn}[multiset operations]
    Given two multisets $S$ and $T$, we define their \emph{disjoint union} $S\uplus T$ by
    \begin{equation*}
        \mathrm{Count}_{S\uplus T}(x)\coloneqq \mathrm{Count}_{S}(x)+\mathrm{Count}_{T}(x),
    \end{equation*} 
    
    Let $n\coloneqq\max(|S|, |T|)$.
    Define their \emph{intersection} $S \cap T$ and \emph{union} $S \cup T$ by
    \begin{align*}
        S \cap T & \coloneqq \{S[1] \wedge T[1], S[2] \wedge T[2], \dots, S[n] \wedge T[n]\}, \\
        S \cup T & \coloneqq \{S[1] \vee T[1], S[2] \vee T[2], \dots, S[n] \vee T[n]\}.
    \end{align*}
    
    We say $S \subseteq T$ if there is some $R$ such that $S \cup R \equiv T$.
\end{defn}

\subsection{Associative Arrays}

In order to perform operations on the entries of a fuzzy database, we introduce associative arrays, a generalization of matrices that allows us to perform generalized linear algebra operations.

\begin{defn}[standard associative array]
    A \emph{standard associative array} (or a \emph{standard array}) is a function $\arr{A} \colon I_\arr{A} \times J_\arr{A} \to \mathbb{V}$, where $I_\arr{A}$ and $J_\arr{A}$ are arbitrary sets known as the \emph{row} and \emph{column supports} of $\arr{A}$, respectively, and $\mathbb{V} = (V, \oplus, \otimes, 0, 1)$ is a semiring. 
\end{defn}

In this paper, we use the standard associative array operations described in \cite{8258298}. These include element-wise products, array multiplication and addition, the identity arrays, transposes, and the array Kronecker products. Note that \cite{8258298} refers to standard associative arrays as simply ``associative arrays''; we include the word ``standard'' to distinguish from fuzzy associative arrays, defined later.

\begin{defn}
    Suppose $f$ is a symmetric function that takes in tuples of nonzero values in $\mathbb{V}$ and outputs a value in $\mathbb{V}$, and $f(\emptyset)=0$. Given standard arrays $\arr{A}$ and $\arr{B}$ such that $J_\arr{A} = I_\arr{B}$, define $\arr{A} \arrprod{f}{\otimes} \arr{B}$ as the standard array such that 
    \begin{align*}
        (\arr{A} \arrprod{f}{\otimes} \arr{B})(k_1, k_2) \coloneqq f\Bigl( \bigl(\arr{A}(k_1, j) \otimes \arr{B}(j, k_2)\bigr)_{j \in J} \Bigr),
    \end{align*} 
    where $J = \{j \in J_\arr{A} \mid \arr{A}(k_1, j_i) \otimes \arr{B}(j_i, k_2) \neq 0\}$.
\end{defn}

\begin{defn}[rows]
    A function $r$ is a \emph{row} of $\arr{A}$ if there is $k_1 \in I_\arr{A}$ such that $r(k_2) = \arr{A}(k_1, k_2)$ for all $k_2 \in J_\arr{A}$. We use the notation $\arr{A}[k_1, \colo]$ to represent $f$.

    Given a row $r$ of $\farray{A}$ and a set of column keys $J$, $r|_J$ is $r$ restricted to the domain $J_\arr{A} \cap J$.
\end{defn}

\begin{defn}[regularization]
    Given a standard array $\arr{A} \colon I_\arr{A} \times J_\arr{A} \rightarrow \mathbb{V}$, the \emph{regularization} of $\arr{A}$ is
    \begin{equation*}
        \omega(\arr{A}) \coloneqq (\mathbb{I}_{I_\arr{A}, \mathbb{V}^{J_\arr{A}}, \arr{A}[-, \colo]})^\intercal \arrprod{\mu}{\otimes} \arr{A},
    \end{equation*}
    where $\mu$ sends a tuple $(v_1, v_2, \ldots, v_n)$ to $v_1$.
\end{defn}

In other words, $\omega(\arr{A})$ effectively removes duplicate rows from $\arr{A}$. This may be made precise in terms of weak and strong equivalence:

\begin{defn}[strong{,} weak equivalence of standard arrays]
    Two standard associative arrays $\arr{A}, \arr{B}$ are \emph{strongly equivalent}, denoted $\arr{A} \equiv \arr{B}$, if there exists a bijection $f \colon I_\arr{A} \rightarrow I_\arr{B}$ such that $\arr{A}(m, n) = \arr{B}(f(m), n)$ for all $m, n$.
    
    $\arr{A}$ and $\arr{B}$ are \emph{weakly equivalent}, denoted $\arr{A} \approx \arr{B}$, if $\omega(\arr{A})\equiv \omega(\arr{B})$.
    
    Note that both strong and weak equivalence are transitive, reflexive, and symmetric.
\end{defn}

\begin{prop}\label{OmegaRows}
A row $r$ is in $\arr{A}$ if and only if it is in $\omega(\arr{A})$, and thus $\omega(\arr{A}) \approx \arr{A}$.
\end{prop}
 
\begin{lemma}\label{OmegaEquiv}
    All the rows in $\omega(\arr{A})$ are distinct. As a corollary, if all the rows of $\arr{A}$ are distinct, then $\omega(\arr{A})\equiv \arr{A}$. Additionally, two standard arrays $\arr{A}$ and $\arr{B}$ are weakly equivalent if and only if they have the same rows (not counting multiplicity).
\end{lemma}
\begin{proof}
    This follows by Proposition~\ref{OmegaRows} and the fact that $\omega(\arr{A})[r, \colo] = \omega(\arr{A})[r', \colo]$ if and only if $r = r'$, since $\omega(\arr{A})[r, \colo] = r$ and $\omega(\arr{A})[r', \colo] = r'$.
\end{proof}

\section{Fuzzy arrays}\label{FArrays}

In this section we combine standard associative arrays with multisets to define fuzzy associative arrays, which are a mathematical representation of fuzzy databases.

\begin{defn}[fuzzy associative array]
    A \emph{fuzzy associative array} (or a \emph{fuzzy array} or an \emph{array}) is an ordered pair $\farray{A}=(\arr{A}, \varphi)$ where $\arr{A}\colon I_\arr{A}\times J_\arr{A}\rightarrow \mathbb{V}$ is a standard associative array, and $\varphi\colon \mathbb{V}^{J_\arr{A}}\rightarrow \{\text{finite multisets of fuzzy values}\}$ is a function (called the \emph{fuzzy component} of $\farray{A}$).

    We define $\farray{A}[k_1, :]\coloneqq \arr{A}[k_1, :]$ for all $k_1\in I_\arr{A}$.
\end{defn}

The reason for defining fuzzy arrays as such is that given an array $\farray{A}$ and a row $r$ in $\farray{A}$, $\varphi_\farray{A}(r)$ is a multiset corresponding to the degree that $r$ belongs to $\farray{A}$. For example, suppose one wished to create a database $\farray{A}$ of all tall people. Then if $r$ is the row containing the single value ``John" corresponding to the column key ``Name", each fuzzy value in the multiset $\varphi_\farray{A}(r)$ would represent a person named John, and the fuzzy value would represent the degree to which that John is tall.

\begin{defn}[zero row]
    A row $r$ of an array $\farray{A}$ is a \emph{zero row} if $r \equiv 0$ or $\varphi_\farray{A}(r) \equiv \emptyset$. $r$ is a \emph{nonzero row} otherwise.
\end{defn}

\begin{remark}\label{EmptyRows}
    By convention, $\varphi_\farray{A}(r) = \emptyset$ whenever $r$ is a row not in $\farray{A}$ or is a zero row.
\end{remark}

Because we use multisets of fuzzy values to encapsulate the degrees of belonging of each row in a fuzzy array, there is no need to have duplicate rows in arrays. As such, we define the following:

\begin{defn}[regular array]
    We say that a fuzzy array $\farray{A}$ is \emph{regular} if all rows are distinct.
\end{defn}

Throughout this paper, we assume arrays are regular unless stated otherwise. A nonregular array $\farray{A}$ can be made regular by replacing $\arr{A}$ with $\omega(\arr{A})$ and $\varphi_\farray{A}$ with $\lambda r \colon \biguplus_{i \in I_\arr{A}, \arr{A}[i] = r}{\varphi_\farray{A}(i)}$.

\begin{defn}[zero array]
    The zero array $\mathbf{0}$ is the array $(\arr{A}, \varphi)$ such that $\arr{A}\colon (\emptyset\times\emptyset)\rightarrow \mathbb{V}$. Consequently, by Remark \ref{EmptyRows}, $\varphi$ maps all rows to $\emptyset$.
\end{defn}

\begin{defn}[sub-array]
    Given two arrays $\farray{A}$ and $\farray{B}$, we say $\farray{A}\subseteq \farray{B}$ if for each row $r$ in $\farray{A}$ or $\farray{B}$, $\varphi_\farray{A}(r)\subseteq \varphi_\farray{B}(r)$.
\end{defn}

\begin{defn}[zero-padding]\label{ZeroPadding}
    Given an array $\farray{A}=(\arr{A}, \varphi_\farray{A})$, a set of row keys $I$, and a set of column keys $J$, we can \emph{zero-pad} $\farray{A}$ with the key sets $I$ and $J$ by defining
    \begin{equation*}
        \mathrm{pad}_{I, J}(\farray{A}) \coloneqq (\arr{A}_{I, J},
        % \varphi_{\mathrm{pad}_{I, J}(\farray{A})}
        \lambda r \colon \varphi_\farray{A}(r|_{J_\arr{A}})
        ),
    \end{equation*}
    where $\arr{A}_{I, J} \coloneqq \mathbb{I}_{I_\arr{A} \cup I, I_\arr{A}} \arrprod{\oplus}{\otimes} \arr{A} \arrprod{\oplus}{\otimes} \mathbb{I}_{J_\arr{A}, J_\arr{A} \cup J}$.
\end{defn}

Throughout this paper, all operations will be done up to appropriate zero-padding. In other words, if operations are carried out on incompatible standard arrays or multisets, we will implicitly zero-pad to ensure the standard arrays or multisets are compatible with the operation.

\begin{defn}[strong{,} weak equivalence of fuzzy arrays]
    Arrays $\farray{A}$ and $\farray{B}$ (assumed to have equal column keys via zero-padding) are \emph{strongly equivalent}, denoted $\farray{A} \equiv \farray{B}$, 
    if there is a partial injection $f \colon I_\arr{A} \rightharpoonup I_\arr{B}$ such that:
    \begin{enumerate}[(i)]
        \item $\arr{A}[i, \colo] = \arr{B}[f(i), \colo]$ for all $i \in \dom(f)$.
        \item $\varphi_\farray{A}(\arr{A}[i, \colo]) \equiv \varphi_\farray{B}(\arr{B}[f(i), \colo])$ for all $i \in \dom(f)$.
        \item $\varphi_\farray{A}(\arr{A}[i, \colo]) \equiv \varphi_\farray{B}(\arr{B}[i', \colo]) \equiv \emptyset$ for all $i \notin \dom(f)$ and $i' \notin \ran(f)$.
    \end{enumerate}

    If (ii) above is replaced by the condition that $\bigvee{\varphi_\farray{A}(\arr{A}[i, \colo])} = \bigvee{\varphi_\farray{B}(\arr{B}[f(i), \colo])}$ for all $i \in \dom(f)$, then $\farray{A}$ and $\farray{B}$ are \emph{weakly equivalent}, denoted $\farray{A} \approx \farray{B}$.

    Note that strong equivalence implies weak equivalence, and both equivalences are transitive, reflexive, and symmetric.
\end{defn}

\begin{remark} 
    For any array $\farray{A}$ and key sets $I$, $J$ we have $\farray{A}\equiv \mathrm{pad}_{I, J}(\farray{A})$. This ensures the relations defined in the sequel are well-defined up to equivalences if we zero-pad.
\end{remark}

\section{Unary Fuzzy Relations}\label{Relations}

In this section we discuss various unary operations on fuzzy arrays and some of their relevant properties.

\begin{defn}[projection]
    The \emph{projection} of $\farray{A}=(\arr{A}, \varphi_\farray{A})$ onto a set of column keys $J$ is defined by 
    \begin{equation*}
        \Pi_J(\farray{A}) \coloneqq \Biggl(\omega(\arr{A} \arrprod{\oplus}{\otimes} \mathbb{I}_J), \lambda r \colon \biguplus_{s, s|_J = r} \varphi_\farray{A}(s) \Biggr).
    \end{equation*}
\end{defn}

\begin{prop}
    $\Pi_\emptyset(\farray{A})\equiv \mathbf{0}$.
\end{prop}

\begin{prop}
    $\Pi_{J_1}(\Pi_{J_2}(\farray{A}))\equiv \Pi_{J_1\cap J_2}(\farray{A})$.
\end{prop}
 
\begin{defn}[selection]
    Let $\psi$ be a \emph{fuzzy condition} (a map from $\mathbb{V}^{J_\arr{A}}$ to $[0, 1]$) on the rows of $\farray{A} = (\arr{A}, \varphi_\farray{A})$. Then the \emph{selection} of $\farray{A}$ based on $\psi$ is 
    \begin{equation*}
        \sigma_{\psi}(\farray{A}) \coloneqq (\arr{A}, \lambda r \colon \{ x \wedge \psi(r) \mid x \in \varphi_\farray{A}(r) \}).
    \end{equation*}
\end{defn}

\begin{prop}
    If $\psi\equiv 1$, then $\sigma_{\psi}$ is the identity map up to strong equivalence. If $\psi\equiv 0$, then $\sigma_{\psi}$ is the zero map up to strong equivalence.
\end{prop}

\begin{prop}
    If $\psi_1$, $\psi_2$ are two fuzzy conditions, then $\sigma_{\psi_1}(\sigma_{\psi_2}(\farray{A}))\equiv\sigma_{\psi_3}(\farray{A})$, where $\psi_3 = \lambda r \colon \psi_1(r)\wedge \psi_2(r)$.
\end{prop}

\begin{defn}[renaming]
    Let $f \colon J \to J'$ be a bijection between sets of column keys. Then given an array $\farray{A} = (\arr{A}, \varphi_\farray{A})$ with $J_\arr{A} \subseteq J$, the \emph{renaming} of $\farray{A}$ via $f$ is 
    \begin{align*}
        \rho_{f}(\farray{A}) &\coloneqq \left(\arr{A} \arrprod{\oplus}{\otimes} \mathbb{I}_{f}, \lambda r \colon \varphi_\farray{A}((r \circ f)|_{J_\arr{A}}) \right).
    \end{align*}
\end{defn}

\begin{prop}
    Let $J$ be a universal set of column keys. Then $\rho_{J, J, \mathrm{id}_J}$ is the identity map up to strong equivalence.
\end{prop}

\section{Binary Fuzzy Relations}\label{Relations2}

In this section we discuss binary operations on fuzzy arrays. By zero-padding, we assume without loss of generality that all fuzzy arrays have the same column keys.

\begin{defn}[union of standard associative arrays]
    The \emph{union} of two standard associative arrays $\arr{A}$ and $\arr{B}$ is 
    \begin{equation*}
        \arr{A}\amalg \arr{B}\coloneqq \omega((\mathbb{I}_{I_\arr{A}\times \{1\}, I_\arr{A}}\arr{A})\oplus (\mathbb{I}_{I_\arr{B}\times \{2\}, I_\arr{B}}\arr{B})).
    \end{equation*}
\end{defn}

\begin{example}
~
\begin{equation*}
    \begin{array}{| c |}
        \hline
        \arr{A} \\ \hline
        \!\!\begin{array}{l | l}
            Name & Age \\ \hline
            John & 30 \\
            Sam & 28
        \end{array}\!\! \\ \hline
    \end{array}
    \amalg
    \begin{array}{| c |}
        \hline
        \arr{B} \\ \hline
        \!\!\begin{array}{l | l}
            Name & Age \\ \hline
            John & 30 \\
            John & 35
        \end{array}\!\! \\ \hline
    \end{array}
    =
    \begin{array}{| c |}
        \hline
        \arr{A} \amalg \arr{B} \\ \hline
        \!\!\begin{array}{l | l}
            Name & Age \\ \hline
            John & 30 \\
            John & 35 \\
            Sam & 28
        \end{array}\!\! \\ \hline
    \end{array} 
\end{equation*}
\end{example}

\begin{defn}[disjoint union of fuzzy arrays]
    Given two arrays $\farray{A}=(\arr{A}, \varphi_\farray{A})$ and $\farray{B}=(\arr{B}, \varphi_\farray{B})$, their \emph{disjoint union} is 
    \begin{equation*}
        \farray{A} \uplus \farray{B} \coloneqq (\arr{A}\amalg \arr{B}, \lambda r \colon \varphi_\farray{A}(r) \uplus \varphi_\farray{B}(r)).
    \end{equation*}
\end{defn}

\begin{defn}[union{,} intersection of fuzzy arrays]
    The \emph{union} of two fuzzy arrays $\farray{A}$ and $\farray{B}$ is 
    \begin{equation*}
        \farray{A} \cup \farray{B} \coloneqq (\arr{A} \amalg \arr{B}, \lambda r \colon \varphi_\farray{A}(r) \cup \varphi_\farray{B}(r)).
    \end{equation*}
    
    Intersections are defined analogously with $\varphi_{\farray{A} \cap \farray{B}} \coloneqq \lambda r \colon \varphi_\farray{A}(r) \cap \varphi_\farray{B}(r)$.
\end{defn}

\begin{defn}[difference of fuzzy arrays]
    The \emph{difference} of two fuzzy arrays $\farray{A}$ and $\farray{B}$ is 
    \begin{equation*}
        \farray{A} \setminus \farray{B} \coloneqq (\arr{A} \amalg \arr{B}, \lambda r \colon \{\varphi_\farray{A}(r)[k] \cdot \delta_{k, \farray{A}, \farray{B}}(r) \mid k \in \mathbb{N}_{>0}\}),
    \end{equation*}
    where for each row $r$ and $k \in \mathbb{N}_{>0}$, $\delta_{k, \farray{A}, \farray{B}}(r) = 1$ if $\varphi_\farray{A}(r)[k]> \varphi_\farray{B}(r)[k]$ and $\delta_{k, \farray{A}, \farray{B}}(r) = 0$ otherwise.
\end{defn}

\begin{lemma}\label{SqcupAss}
    $\amalg$ is associative up to strong equivalence of standard arrays.
\end{lemma}
\begin{proof}
    Take any standard arrays $\arr{A}, \arr{B}, \arr{C}$. Let 
    \begin{equation*}
        \arr{D} \coloneqq \arr{A} \amalg \arr{B} = \omega((\mathbb{I}_{I_\arr{A} \times \{1\}} \arrprod{\oplus}{\otimes} \arr{A}) \oplus (\mathbb{I}_{I_\arr{B} \times \{2\}} \arrprod{\oplus}{\otimes} \arr{B})).
    \end{equation*}
    Then 
    \begin{align*}
        (\arr{A} \amalg \arr{B}) \amalg \arr{C} & = \omega((\mathbb{I}_{I_\arr{D} \times \{1\}} \arrprod{\oplus}{\otimes} \arr{D}) \oplus (\mathbb{I}_{I_\arr{C} \times \{2\}} \arrprod{\oplus}{\otimes} \arr{C}) \\
        & \approx (\mathbb{I}_{I_\arr{D} \times \{1\}} \arrprod{\oplus}{\otimes} \arr{D}) \oplus (\mathbb{I}_{I_\arr{C} \times \{2\}} \arrprod{\oplus}{\otimes} \arr{C})
    \end{align*}
    by Proposition~\ref{OmegaRows}. We also have 
    \begin{align*}
        \mathbb{I}_{I_\arr{D}\times \{1\}} \arrprod{\oplus}{\otimes} \arr{D} & \equiv \arr{D} \\
        & \approx (\mathbb{I}_{I_\arr{A} \times \{1\}} \arrprod{\oplus}{\otimes} \arr{A}) \oplus (\mathbb{I}_{I_\arr{B} \times \{2\}} \arrprod{\oplus}{\otimes} \arr{B}) \\
        & \equiv (\mathbb{I}_{I_\arr{A} \times \{1\}} \arrprod{\oplus}{\otimes} \arr{A}) \oplus (\mathbb{I}_{I_\arr{B} \times \{3\}} \arrprod{\oplus}{\otimes} \arr{B}).
    \end{align*} 
    Since $\mathbb{I}_{I_\arr{D} \times \{1\}} \arrprod{\oplus}{\otimes} \arr{D}$ and $\mathbb{I}_{I_\arr{C}\times \{2\}} \arrprod{\oplus}{\otimes} \arr{C}$ share no row keys, and neither do $(\mathbb{I}_{I_\arr{A} \times \{1\}} \arrprod{\oplus}{\otimes} \arr{A}) \oplus (\mathbb{I}_{I_\arr{B} \times \{3\}} \arrprod{\oplus}{\otimes} \arr{B})$ and $\mathbb{I}_{I_\arr{C} \times \{2\}} \arrprod{\oplus}{\otimes} \arr{C}$, this means 
    \begin{align*}
        & (\mathbb{I}_{I_\arr{D} \times \{1\}} \arrprod{\oplus}{\otimes} \arr{D}) \oplus (\mathbb{I}_{I_\arr{C} \times \{2\}} \arrprod{\oplus}{\otimes} \arr{C}) \\
        & \approx (\mathbb{I}_{I_\arr{A} \times \{1\}} \arrprod{\oplus}{\otimes} \arr{A}) \oplus (\mathbb{I}_{I_\arr{B} \times \{3\}} \arrprod{\oplus}{\otimes} \arr{B}) \oplus (\mathbb{I}_{I_\arr{C} \times \{2\}} \arrprod{\oplus}{\otimes} \arr{C}).
    \end{align*}
     
    This implies $(\arr{A}\amalg \arr{B})\amalg \arr{C} \approx (\arr{B}\amalg \arr{C})\amalg \arr{A}\equiv \arr{A}\amalg(\arr{B}\amalg \arr{C})$ via an appropriate renaming of the row keys. But by Proposition~\ref{OmegaRows}, $(\mathbb{I}_{I_{(\arr{A} \amalg \arr{B}) \times \{1\}}} \arrprod{\oplus}{\otimes} (\arr{A}\amalg \arr{B})) \oplus (\mathbb{I}_{I_\arr{C} \times \{2\}} \arrprod{\oplus}{\otimes} \arr{C}) \approx (\arr{A} \amalg \arr{B}) \amalg \arr{C}$ and similarly $(\mathbb{I}_{I_\arr{A} \times \{1\}} \arrprod{\oplus}{\otimes} \arr{A}) \oplus (\mathbb{I}_{I_{(\arr{B} \amalg \arr{C})} \times \{2\}} \arrprod{\oplus}{\otimes} (\arr{B} \amalg \arr{C})) \approx \arr{A} \amalg (\arr{B} \amalg \arr{C})$, so 
    \begin{align*}
        & \mathbb{I}_{I_{(\arr{A} \amalg \arr{B}) \times\{1\}}} \arrprod{\oplus}{\otimes} (\arr{A} \amalg \arr{B}) \oplus \mathbb{I}_{I_\arr{C} \times \{2\}} \\
        & \approx \mathbb{I}_{I_\arr{A} \times \{1\}} \arrprod{\oplus}{\otimes} \arr{A} \oplus \mathbb{I}_{I_{(\arr{B} \amalg \arr{C})} \times \{2\}} \arrprod{\oplus}{\otimes} (\arr{B} \amalg \arr{C}).
    \end{align*}
    By Lemma~\ref{OmegaEquiv} this means $(\arr{A} \amalg \arr{B}) \amalg \arr{C} \equiv \arr{A} \amalg (\arr{B} \amalg \arr{C})$.
\end{proof}

\begin{theorem}
    $\uplus$, $\cup$, and $\cap$ are commutative and associative up to strong equivalence of fuzzy arrays.
\end{theorem}
\begin{proof}
    Define $f \colon (I_\arr{A} \times \{1\}) \cup (I_\arr{B} \times \{2\}) \rightarrow (I_\arr{B} \times \{1\}) \cup (I_\arr{A} \times \{2\})$ by $f((i, 1)) \coloneqq (i, 2)$ and $f((i', 2)) \coloneqq (i', 1)$ for any $i \in I_\arr{A}$ and $i' \in I_\arr{B}$. Then note 
    \begin{align*}
        & ((\mathbb{I}_{I_\arr{A} \times \{1\}, I_\arr{A}} \arrprod{\oplus}{\otimes} \arr{A}) \oplus (\mathbb{I}_{I_\arr{B} \times \{2\}, I_\arr{B}} \arrprod{\oplus}{\otimes} \arr{B}))(m, n) \\
        & = ((\mathbb{I}_{I_\arr{B} \times \{1\}, I_\arr{B}} \arrprod{\oplus}{\otimes} \arr{B}) \oplus (\mathbb{I}_{I_\arr{A} \times \{2\}, I_\arr{A}} \arrprod{\oplus}{\otimes} \arr{A}))(f(m), n).
    \end{align*} 
    Then this implies $\arr{A} \amalg \arr{B} \equiv \arr{B} \amalg \arr{A}$. From here, it is clear that all three relations are commutative, since the fuzzy components of $\farray{A} \uplus \farray{B}$, $\farray{A} \cup \farray{B}$, and $\farray{A} \cap \farray{B}$ are commutative (by the commutativity of $\cup$, $\wedge$, and $\vee$).
    
    Lemma~\ref{SqcupAss} combined with the associativity of $\cup$, $\wedge$, and $\vee$ implies the associativity of the three given results. 
\end{proof}

\begin{theorem}
    $\cap$ and $\cup$ distribute over $\uplus$ up to strong equivalence of fuzzy arrays.
\end{theorem}
\begin{proof}
    We show $\farray{A} \uplus (\farray{B} \cap \farray{C}) \equiv (\farray{A} \uplus \farray{B}) \cap (\farray{A} \uplus \farray{C})$. Let $\farray{X} \coloneqq \farray{A} \uplus (\farray{B} \cap \farray{C})$ and $\farray{Y} \coloneqq (\farray{A} \uplus \farray{B}) \cap (\farray{A} \uplus \farray{C})$. By the commutativity and associativity of $\amalg$ up to strong equivalence, $\arr{A} \amalg (\arr{B} \amalg \arr{C}) \equiv (\arr{A} \amalg \arr{B}) \amalg (\arr{A} \amalg \arr{C})$.
    
    Now it suffices to show that $\varphi_\farray{X}(r) \equiv \varphi_\farray{Y}(r)$ for any row $r$ in $\arr{A} \amalg (\arr{B} \amalg \arr{C})$. For any $c \in [0, 1]$ and fuzzy associative array $\farray{Z}$, define $\varphi_\farray{Z}(r)_{\geq c}$ to be the number of elements in $\varphi_\farray{Z}(r)$ that are at least $c$. Define $\varphi_\farray{Z}(r)_{>c}$ similarly. Then note 
    \begin{align*}
        & \varphi_\farray{X}(r)[k] \geq c \\
        & \iff \varphi_\farray{A}(r)_{\geq c} + \min(\varphi_\farray{B}(r)_{\geq c}, \varphi_\farray{C}(r)_{\geq c}) \geq k \\
        & \iff \min(\varphi_\farray{A}(r)_{\geq c} + \varphi_\farray{B}(r)_{\geq c}, \varphi_\farray{A}(r)_{\geq c} + \varphi_\farray{C}(r)_{\geq c}) \geq k \\
        & \iff \varphi_\farray{Y}(r)[k] \geq c.
    \end{align*}
    
    We can similarly show $\varphi_\farray{X}(r)[k]\leq c \iff \varphi_\farray{Y}(r)[k] \leq c$. This implies $\varphi_\farray{X}(r)[k] = \varphi_\farray{Y}(r)[k]$ for all $k$, so $\varphi_\farray{X}(r) \equiv \varphi_\farray{Y}(r)$ as desired.
    
    By commutativity, $\cap$ is also right-distributive over $\uplus$, and an analogous proof shows $\cup$ distributes over $\uplus$.
\end{proof}.

Note $\uplus$ does not distribute over $\cap$ and $\cup$. For example, consider the arrays $\farray{A}$, $\farray{B}$, and $\farray{C}$ each containing the single row $r$. Then suppose $\varphi_\farray{A}(r) = \{1, 0\}$, $\varphi_\farray{B}(r) = \{1, 1\}$, and $\varphi_\farray{C}(r) = \{1, 1\}$. Then $\farray{A} \cap (\farray{B} \uplus \farray{C}) \not\equiv (\farray{A} \cap \farray{B}) \uplus (\farray{A} \cap \farray{C})$

To show $\uplus$ doesn't distribute over $\cup$ instead take $\varphi_\farray{A}(r)=\{1, 1\}$, $\varphi_\farray{B}(r)=\{1, 0\}$, and $\varphi_\farray{C}(r)=\{1, 0\}$. 

\begin{prop}
    $\farray{A} \setminus \farray{A} \equiv \mathbf{0}$ for all arrays $\farray{A}$.
\end{prop}

\begin{theorem}\label{DifferenceDist}
    $(\farray{A}\setminus\farray{C})\cup (\farray{B}\setminus\farray{C})\subseteq (\farray{A}\cup \farray{B})\setminus\farray{C}$.
\end{theorem}
\begin{proof} 
    Let $\farray{X}=(\farray{A}\setminus\farray{C})\cup (\farray{B}\setminus\farray{C})$ and $\farray{Y}=(\farray{A}\cup \farray{B})\setminus\farray{C}$. By the associativity and commutativity of $\amalg$, $\amalg$ distributes over itself, so it suffices to check $\varphi_\farray{X}(r)\subseteq \varphi_\farray{Y}(r)$ for any row $r$. 
    
    Consider some row $r$, and let $m$ be the maximum cardinality among $\varphi_\farray{A}(r)$, $\varphi_\farray{B}(r)$, and $\varphi_\farray{C}(r)$. Then 
    \begin{align*}
        \varphi_{\farray{A}\setminus\farray{C}}(r) & \equiv \{\varphi_\farray{A}(r)[k]\cdot \delta_{k, \farray{A}, \farray{C}}(r)\mid k\leq m\}, \\
        \varphi_{(\farray{A}\cup\farray{B}) \setminus \farray{C}}(r) & \equiv \{\varphi_{\farray{A} \cup \farray{B}}(r)[k] \cdot \delta_{k, \farray{A} \cup \farray{B}, \farray{C}}(r) \mid k \leq m\}.
    \end{align*}
    
    Now note for any fixed $k$, $\varphi_\farray{A}(r)[k]\leq \varphi_{\farray{A}\cup\farray{B}}(r)[k]$ and $\delta_{k, \farray{A}, \farray{C}}(r)\leq \delta_{k, \farray{A}\cup\farray{B}, \farray{C}}(r)$. This implies that $\varphi_{\farray{A}\setminus\farray{C}}(r)\subseteq \varphi_{(\farray{A}\cup\farray{B})\setminus\farray{C}}(r)$, because the above inequalities imply that among the $m$ elements (after zero-padding appropriately) of these two multisets, we can construct a bijection $f\colon\varphi_{\farray{A}\setminus\farray{C}}(r)\rightarrow  \varphi_{(\farray{A}\cup\farray{B})\setminus\farray{C}}(r)$ such that $x\leq f(x)$ for all $x\in \varphi_{\farray{A}\setminus\farray{C}}(r)$. In particular, this implies that the $\varphi_{\farray{A}\setminus\farray{C}}(r)[k]\leq \varphi_{(\farray{A}\cup\farray{B})\setminus\farray{C}}(r)[k]$, which is equivalent to $\varphi_{\farray{A}\setminus\farray{C}}(r)\subseteq \varphi_{(\farray{A}\cup\farray{B})\setminus\farray{C}}(r)$.
    
    Analogous reasoning shows $\varphi_{\farray{B}\setminus\farray{C}}(r)\subseteq \varphi_{(\farray{A}\cup\farray{B})\setminus\farray{C}}(r)$, which is enough to imply $\varphi_\farray{X}(r)\subseteq \varphi_\farray{Y}(r)$.
\end{proof}

\begin{remark}
    With traditional sets, Theorem \ref{DifferenceDist} can be strengthened to replace $\subseteq$ with equality. However, due to the more general nature of fuzzy multisets, equality need not hold. 

    One can similarly show $(\farray{A}\setminus\farray{B})\cup (\farray{A}\setminus\farray{C})\subseteq  \farray{A}\setminus(\farray{B}\cap\farray{C})$.
\end{remark}

\begin{theorem}\label{ProjWeakEquiv}
    $\Pi_J$ preserves $\cup$ up to weak equivalence.
\end{theorem}
\begin{proof} 
    Note that the rows of $\omega(\arr{A} \arrprod{\oplus}{\otimes} \mathbb{I}_J)$ are the rows of $\arr{A}$ restricted to $J$ (not including multiplicity) and similarly the rows of $\omega(\arr{B} \arrprod{\oplus}{\otimes} \mathbb{I}_J)$ are the rows of $\arr{B}$ restricted to $J$. Now by Lemma \ref{OmegaRows}, the rows of $\omega(\arr{A} \arrprod{\oplus}{\otimes} \mathbb{I}_J) \amalg \omega(\arr{B} \arrprod{\oplus}{\otimes} \mathbb{I}_J)$ are the rows of both $\arr{A}$ and $\arr{B}$ restricted to $J$.
    
    But note these are also the rows of $(\arr{A} \amalg \arr{B}) \arrprod{\oplus}{\otimes} \mathbb{I}_J$, so in fact $(\arr{A} \amalg \arr{B}) \arrprod{\oplus}{\otimes} \mathbb{I}_J \approx \omega(\arr{A} \arrprod{\oplus}{\otimes} \mathbb{I}_J) \amalg \omega(\arr{B} \arrprod{\oplus}{\otimes} \mathbb{I}_J)$. But $\omega((\arr{A} \amalg \arr{B}) \arrprod{\oplus}{\otimes} \mathbb{I}_J) \approx (\arr{A} \amalg \arr{B}) \arrprod{\oplus}{\otimes} \mathbb{I}_J$ by Lemma~\ref{OmegaRows}, so $\omega((\arr{A} \amalg \arr{B}) \arrprod{\oplus}{\otimes} \mathbb{I}_J) \approx \omega(\arr{A} \arrprod{\oplus}{\otimes} \mathbb{I}_J) \amalg \omega(\arr{B} \arrprod{\oplus}{\otimes} \mathbb{I}_J)$. 
    
    The result follows by 
    \begin{align*}
        \bigvee_{x \in \varphi_{\Pi_J(\arr{A} \cup \arr{B})}(r)} \hspace{-1.3mm} {x} & = \bigvee_{y|_J = r} \Bigg(\bigvee_{x \in \varphi_\farray{A}(y) \cup \varphi_\farray{B}(y)} {x}\Bigg) \\
        & = \bigvee_{x \in \varphi_{\Pi_J(\arr{A}) \cup \Pi_J(\arr{B})}(r)}{x}. \qedhere
    \end{align*}
\end{proof}

\begin{remark}
    Note that $\Pi_J$ does not preserve $\cup$ up to strong equivalence. For example, consider an array $\farray{A}$ with two rows $r$ and $t$. Suppose $r|_J=t|_J=s$ and take $\varphi_\farray{A}(r)=\{0.5, 0.5\}$, $\varphi_\farray{A}(t)=\{1\}$. Then take $\farray{B}$ with the same two rows, and $\varphi_\farray{B}(r)=\{1\}$, $\varphi_\farray{B}(t)=\{0.5, 0.5\}$. Then $\varphi_{\Pi_J(\farray{A}\cup\farray{B})}(s)\equiv \{1, 1, 0.5, 0.5\}$, but $\varphi_{\Pi_J(\farray{A})\cup\Pi_J(\farray{B})}(s)\equiv \{1, 0.5, 0.5\}$ so $\Pi_J(\farray{A}\cup\farray{B})\not\equiv \Pi_J(\farray{A})\cup\Pi_J(\farray{B})$.  

    $\Pi_J$ also does not preserve $\cap$ or $\setminus$ even up to weak equivalence. Using the same arrays $\farray{A}$ and $\farray{B}$, $\varphi_{\Pi_J(\farray{A}\cap\farray{B})}(s)\equiv\{0.5, 0.5\}$, $\varphi_{\Pi_J(\farray{A})\cap\Pi_J(\farray{B})}(s)\equiv\{1, 0.5, 0.5\}$, $\varphi_{\Pi_J(\farray{A}\setminus\farray{B})}(s)\equiv \{1, 0.5\}$, and $\varphi_{\Pi_J(\farray{A})\setminus\Pi_J(\farray{B})}(s)\equiv \emptyset$. However, $\Pi_J$ does preserve $\uplus$ up to strong equivalence.
\end{remark}

\begin{defn}[theta-join]
    Given a standard array $\arr{A}$, let $I_\arr{A}^{0}$ denote the set of row keys of $\arr{A}$ corresponding to nonzero rows.

    Then if $J_1$ and $J_2$ are sets of column keys, the \emph{theta-join} of two arrays $\farray{A}$ and $\farray{B}$ via a function $\theta \colon \mathbb{V}^{J_1} \times \mathbb{V}^{J_2} \rightarrow [0, 1]$ is 
    \begin{align*}
        \farray{A} \bowtie_\theta \farray{B} & \coloneqq (\arr{A} \bowtie_\theta \arr{B}, \varphi_{\farray{A} \bowtie_\theta \farray{B}}), \\
        \varphi_{\farray{A} \bowtie_\theta \farray{B}}(r) & \coloneqq \{ a \wedge b \wedge \theta(r_A, r_B) \mid a \in \varphi_\farray{A}(g), b \in \varphi_\farray{B}(h)\}, \\
        \arr{A} \bowtie_\theta \arr{B} & \coloneqq (\arr{A} \tensor \mathbb{I}_{I_\arr{B}^0, \{1\}}) \oplus \rho_{\{2\} \times J_\arr{B}, J_\arr{B} \times \{2\}}(\mathbb{I}_{I_\arr{A}^0, \{2\}} \tensor \arr{B}).
    \end{align*}

    Here $r_A(k_2) = r((k_2, 1))$ for $k_2\in J_1$, $r_B(k_2) = r((k_2, 2))$ for $k_2\in J_2$, $g(k) \coloneqq r|_{J_\arr{A} \times\{1\}}((k, 1))$, and $h(k) \coloneqq r|_{J_\arr{B} \times \{2\}}((k, 1))$.
\end{defn}

\begin{example}
    Let $\farray{A}$ and $\farray{B}$ be the following (note row keys are omitted): 
    
    \begin{center}
    \begin{tabular}{|l|l|l|}
    \hline
        \multicolumn{3}{|c|}{$\farray{A}$} \\
        \hline
        Name & Age & $\varphi_\farray{A}$\\
        \hline
        John & 30 & \{1, 0.8\}\\
        Sam & 28 & \{0.9\} \\
        \hline   
    \end{tabular} \hspace{1cm}
    \begin{tabular}{|l|l|l|}
    \hline
        \multicolumn{3}{|c|}{$\farray{B}$} \\
        \hline
        Name & Age & $\varphi_\farray{B}$ \\
        \hline
        Alex & 30 & \{0.6\} \\
        John & 29 & \{0.8\} \\
        \hline   
    \end{tabular}
    \end{center}
    
    Now suppose $\theta$ takes $\arr{A}.Age$ and $\arr{B}.Age$ as its inputs, and yields $1$ if the two ages are equal, $0.5$ if the two ages differ by $1$, and $0$ otherwise. Then $\farray{A}\bowtie_\theta\farray{B}$ is (up to equivalence)
        
    \begin{center}
    \begin{tabular}{|l|l|l|l|l|}
    \hline
        \multicolumn{5}{|c|}{$\farray{A}\bowtie_\theta \farray{B}$} \\
        \hline
        (Name, 1) & (Age, 1) & (Name, 2) & (Age, 2) & $\varphi$ \\
        \hline
        John & 30 & Alex & 30 & \{0.6, 0.6\} \\
        John & 30 & John & 29 & \{0.5, 0.5\} \\
        Sam & 28 & John & 29 & \{0.5\} \\
        Sam & 28 & Alex & 30 & \{\} \\
        \hline   
    \end{tabular}
    \end{center}
    
\end{example}

\begin{prop}
    All relations described in \S \ref{Relations} and \S \ref{Relations2} are preserved up to equivalence upon zero-padding. This allows us to implicitly zero-pad otherwise incompatible arrays before carrying out operations.
\end{prop}

\section{Conclusion}\label{Conclusion}

Fuzzy databases offer a more flexible management system than traditional RDBMSs, allowing for the imprecise or vague language present in our everyday language, and providing applications in areas such as artificial intelligence or control systems. By providing a mathematical framework unifying fuzzy databases in the language of semiring linear algebra and multiset operations, we create a natural and rigorous way to describe these databases.

Further works could include implementing a fuzzy database using the methods described in this paper and analyzing the performance compared to other methods, such as the work in \cite{9166621}. Additionally, one could seek to generalize the work further by replacing fuzzy membership values from the interval $[0, 1]$, to arbitrary residuated lattices. 

\section*{Acknowledgements}

The authors wish to acknowledge the following individuals for their contributions and support: W. Arcand, W. Bergeron, D. Bestor, C. Birardi, B. Bond, S. Buckley, C. Byun, G. Floyd, V. Gadepally, D. Gupta, M. Houle, M. Hubbell, M. Jones, A. Klien, C. Leiserson, K. Malvey, P. Michaleas, C. Milner, S. Mohindra, L. Milechin, J. Mullen, R. Patel, S. Pentland, C. Prothmann, A. Prout, A. Reuther, A. Rosa , J. Rountree, D. Rus, M. Sherman, C. Yee.

\bibliographystyle{IEEEtran}
\bibliography{bibliography}

% Generated by IEEEtran.bst, version: 1.14 (2015/08/26)
\begin{thebibliography}{10}
\providecommand{\url}[1]{#1}
\csname url@samestyle\endcsname
\providecommand{\newblock}{\relax}
\providecommand{\bibinfo}[2]{#2}
\providecommand{\BIBentrySTDinterwordspacing}{\spaceskip=0pt\relax}
\providecommand{\BIBentryALTinterwordstretchfactor}{4}
\providecommand{\BIBentryALTinterwordspacing}{\spaceskip=\fontdimen2\font plus
\BIBentryALTinterwordstretchfactor\fontdimen3\font minus
  \fontdimen4\font\relax}
\providecommand{\BIBforeignlanguage}[2]{{%
\expandafter\ifx\csname l@#1\endcsname\relax
\typeout{** WARNING: IEEEtran.bst: No hyphenation pattern has been}%
\typeout{** loaded for the language `#1'. Using the pattern for}%
\typeout{** the default language instead.}%
\else
\language=\csname l@#1\endcsname
\fi
#2}}
\providecommand{\BIBdecl}{\relax}
\BIBdecl

\bibitem{10.1145/362384.362685}
\BIBentryALTinterwordspacing
E.~F. Codd, ``A relational model of data for large shared data banks,''
  \emph{Commun. ACM}, vol.~13, no.~6, p. 377–387, jun 1970. [Online].
  Available: \url{https://doi.org/10.1145/362384.362685}
\BIBentrySTDinterwordspacing

\bibitem{10.1145/166635.166656}
\BIBentryALTinterwordspacing
H.~Lu, H.~C. Chan, and K.~K. Wei, ``A survey on usage of sql,'' \emph{SIGMOD
  Rec.}, vol.~22, no.~4, p. 60–65, dec 1993. [Online]. Available:
  \url{https://doi.org/10.1145/166635.166656}
\BIBentrySTDinterwordspacing

\bibitem{366566}
P.~Bosc and O.~Pivert, ``Sqlf: a relational database language for fuzzy
  querying,'' \emph{IEEE Transactions on Fuzzy Systems}, vol.~3, no.~1, pp.
  1--17, 1995.

\bibitem{SOSNOWSKI199023}
\BIBentryALTinterwordspacing
Z.~A. Sosnowski, ``Flisp — a language for processing fuzzy data,''
  \emph{Fuzzy Sets and Systems}, vol.~37, no.~1, pp. 23--32, 1990. [Online].
  Available:
  \url{https://www.sciencedirect.com/science/article/pii/016501149090060J}
\BIBentrySTDinterwordspacing

\bibitem{MEDINA199487}
\BIBentryALTinterwordspacing
J.~M. Medina, O.~Pons, and M.~A. Vila, ``Gefred: A generalized model of fuzzy
  relational databases,'' \emph{Information Sciences}, vol.~76, no.~1, pp.
  87--109, 1994. [Online]. Available:
  \url{https://www.sciencedirect.com/science/article/pii/0020025594900698}
\BIBentrySTDinterwordspacing

\bibitem{4052712}
R.~Belohlavek and V.~Vychodil, ``Codd's relational model of data and fuzzy
  logic: Comparisons, observations, and some new results,'' in \emph{2006
  International Conference on Computational Inteligence for Modelling Control
  and Automation and International Conference on Intelligent Agents Web
  Technologies and International Commerce (CIMCA'06)}, 2006, pp. 70--70.

\bibitem{LiteratureOverview}
Z.~Ma, ``A literature overview of fuzzy database modeling,'' \emph{Intelligent
  Databases: Technologies and Applications}, 01 2006.

\bibitem{ZADEH1965338}
\BIBentryALTinterwordspacing
L.~Zadeh, ``Fuzzy sets,'' \emph{Information and Control}, vol.~8, no.~3, pp.
  338--353, 1965. [Online]. Available:
  \url{https://www.sciencedirect.com/science/article/pii/S001999586590241X}
\BIBentrySTDinterwordspacing

\bibitem{LEE1969421}
\BIBentryALTinterwordspacing
E.~Lee and L.~Zadeh, ``Note on fuzzy languages,'' \emph{Information Sciences},
  vol.~1, no.~4, pp. 421--434, 1969. [Online]. Available:
  \url{https://www.sciencedirect.com/science/article/pii/0020025569900255}
\BIBentrySTDinterwordspacing

\bibitem{FuzzyAI}
E.~P. Klement and W.~Slany, \emph{Fuzzy Logic in Artificial Intelligence}, 07
  1997, pp. 179--190.

\bibitem{electronics10222878}
\BIBentryALTinterwordspacing
M.~Ivanova, P.~Petkova, and N.~Petkov, ``Machine learning and fuzzy logic in
  electronics: Applying intelligence in practice,'' \emph{Electronics},
  vol.~10, no.~22, 2021. [Online]. Available:
  \url{https://www.mdpi.com/2079-9292/10/22/2878}
\BIBentrySTDinterwordspacing

\bibitem{8258298}
H.~Jananthan, Z.~Zhou, V.~Gadepally, D.~Hutchison, S.~Kim, and J.~Kepner,
  ``Polystore mathematics of relational algebra,'' in \emph{2017 IEEE
  International Conference on Big Data (Big Data)}, 2017, pp. 3180--3189.

\bibitem{RIESGO2018107}
\BIBentryALTinterwordspacing
{\'{A}}.~Riesgo, P.~Alonso, I.~D{\'{i}}az, and S.~Montes, ``Basic operations
  for fuzzy multisets,'' \emph{International Journal of Approximate Reasoning},
  vol. 101, pp. 107--118, 2018. [Online]. Available:
  \url{https://www.sciencedirect.com/science/article/pii/S0888613X18303797}
\BIBentrySTDinterwordspacing

\bibitem{UnionIntersections}
\BIBentryALTinterwordspacing
------, ``General definitions for the union and intersection of ordered fuzzy
  multisets,'' \emph{Iranian Journal of Fuzzy Systems}, vol.~17, no.~4, pp.
  41--54, 2020. [Online]. Available:
  \url{https://ijfs.usb.ac.ir/article_5405.html}
\BIBentrySTDinterwordspacing

\bibitem{inproceedings}
S.~Miyamoto, ``Fuzzy multisets and their generalizations,'' 08 2000, pp.
  225--236.

\bibitem{9166621}
G.~Manogaran, P.~M. Shakeel, S.~Baskar, C.-H. Hsu, S.~N. Kadry,
  R.~Sundarasekar, P.~M. Kumar, and B.~A. Muthu, ``Fdm: Fuzzy-optimized data
  management technique for improving big data analytics,'' \emph{IEEE
  Transactions on Fuzzy Systems}, vol.~29, no.~1, pp. 177--185, 2021.

\end{thebibliography}

\end{document}